\documentclass[11pt]{article}
\usepackage{fullpage,amsmath,amsthm}
\usepackage[T1]{fontenc}
\usepackage{graphicx}
\usepackage{amssymb,amsmath}
\usepackage{lineno}
\RequirePackage{doi}
\usepackage{hyperref}

\newcommand{\etal}{{et~al.}}
\newcommand{\ie}{{i.e.}}
\newcommand{\eg}{{e.g.}}

\newcommand{\geod}{{\rm geod}}
\newcommand{\greedy}{{\rm greedy}}
\newcommand{\conv}{{\rm conv}}

\newcommand{\per}{{\rm per}}
\newcommand{\dist}{{\rm dist}}
\newcommand{\diam}{{\rm diam}}

\newcommand{\width}{{\rm width}}

\newcommand{\NN}{\mathbb{N}} 
\newcommand{\RR}{\mathbb{R}} 
\newcommand{\eps}{\varepsilon}

\def\A{\mathcal A}

\def\C{\mathcal C}
\def\O{\mathcal O}

\newcommand{\later}[1]{}
\newcommand{\old}[1]{}

\newtheorem{theorem}{Theorem}
\newtheorem{lemma}{Lemma}

\newtheorem{corollary}{Corollary}
\newtheorem{proposition}{Proposition}

\begin{document}

\title{Maximal Distortion of Geodesic Diameters\\ in Polygonal Domains\footnote{A preliminary version of this paper appeared in the \emph{Proceedings of the 34th International Workshop on Combinatorial Algorithms (IWOCA)}, LNCS~13889, Springer, 2023, pp.~197--208.}}

\author{Adrian Dumitrescu\thanks{Algoresearch L.L.C., Milwaukee, WI, USA, and 
Research Institute of the University of Bucharest, Romania, and 
Alfr\'ed R\'enyi Institute of Mathematics, Budapest, Hungary. 
Email: \texttt{ad.dumitrescu@algoresearch.org}}
\and
Csaba D. T\'oth\thanks{California State University Northridge, Los Angeles, CA and Tufts University, Medford, MA, USA. Email: \texttt{csaba.toth@csun.edu}}
}
\date{}
\maketitle              

\begin{abstract}
For a polygon $P$ with holes in the plane, we denote by $\varrho(P)$ the ratio between the geodesic and the Euclidean diameters of $P$. It is shown that over all convex polygons with 
$h$~convex holes, the supremum of $\varrho(P)$ is between $\Omega(h^{1/3})$ and $O(h^{1/2})$. 
The upper bound improves to $\varrho(P)\leq O(1+\min\{h^{3/4}\Delta,h^{1/2}\Delta^{1/2}\})$
if the Euclidean diameter of every hole is most $\Delta$ times the Euclidean diameter of $P$; and to $O(1)$ if every hole is a \emph{fat} convex polygon.
Furthermore, we show that the function $g(h)=\sup_P \varrho(P)$ over convex polygons with $h$ convex holes has the same growth rate as an analogous quantity over geometric triangulations with $h$ vertices when $h\rightarrow \infty$. 
%
\end{abstract}

\section{Introduction} \label{sec:intro}

Determining the maximum distortion between two metrics on the same ground set is a fundamental problem in metric geometry. In this paper, we study the maximum ratio between the geodesic (i.e., shortest path) diameter and the Euclidean diameter over polygons with holes. 
A \emph{polygon $P$ with $h$ holes} (also known as a \emph{polygonal domain}) is defined as follows. Let $P_0$ be a simple polygon, and let $P_1,\ldots , P_h$ be pairwise disjoint simple polygons in the interior of $P_0$. Then $P=P_0\setminus \left(\bigcup_{i=1}^h P_i\right)$. 

The Euclidean distance between two points $s,t\in P$ is $|st|=\|s-t\|_2$, and the shortest path distance $\geod(s,t)$ is the minimum arclength of a polygonal path between $s$ and $t$ contained in $P$. The triangle inequality implies that $|st|\leq \geod(s,t)$ for all $s,t\in P$. The \emph{geometric dilation} of $P$ (also known as the \emph{stretch factor}) is $\sup_{s,t\in P} \geod(s,t)/|st|$. The geometric dilation of $P$ can be  arbitrarily large, even if $P$ is a (nonconvex) quadrilateral. 

The \emph{Euclidean diameter} of $P$ is $\diam_2(P) = \sup_{s,t\in P}|st|$ and its \emph{geodesic diameter} is $\diam_g(P) =  \sup_{s,t\in P} \geod(s,t)$. 
It is clear that $\diam_2(P)\leq \diam_g(P)$. We are interested in the \emph{distortion}
\begin{equation}\label{eq:rho}
\varrho(P)=\frac{\diam_g(P)}{\diam_2(P)}.
\end{equation}
Note that $\varrho(P)$ is unbounded, even for simple polygons. In fact, it is easy to show that $\varrho(P)\leq O(n)$ for every simple polygon $P$ with $n$ vertices, and this bound is the best possible. For a lower bound, consider a zig-zag polygon $P$, bounded by two $x$-monotone polygonal chains with unit-length edges and slopes roughly $\pm n$. Then $\diam_g(P) = \Omega(n)$, and since $P$ fits in an axis-aligned unit square, then we have $\diam_2(P)\leq O(1)$, which yields $\varrho(P)\geq \Omega(n)$. For the upper bound, note that the perimeter of a simple polygon with $n$ vertices is bounded by $|\partial P|\leq n\cdot \diam_2(P)$. For every $s,t\in P$, we can construct an $st$-path of length at most $2\,\diam_2(P)+\frac12\, |\partial P|\leq O(n\cdot \diam_2(P))$ by connecting each point to a closest point on the boundary, and these closest points are connected along $\partial P$.

In this paper, we consider convex polygons with convex holes. Specifically, let $\C(h)$ denote the family of polygonal domains $P=P_0\setminus \left(\bigcup_{i=1}^h P_i\right)$, where $P_0, P_1,\ldots , P_h$ are convex polygons;  and let 
\begin{equation}\label{eq:g}
g(h)=\sup_{P\in \C(h)} \varrho(P).
\end{equation}
It is clear that if $h=0$, then $\geod(s,t)=|st|$ for all $s,t\in P$, which implies $g(0)=1$. Our main result is the following.
\begin{theorem} \label{thm:main}
For every $h\in \NN$, we have $\Omega(h^{1/3})\leq g(h)\leq O(h^{1/2})$.
\end{theorem}

The lower bound construction is a polygonal domain in which all $h$ holes have about the same diameter $\Theta(h^{-1/3})\cdot \diam_2(P)$. We prove a matching upper bound for all polygons $P$ with holes of diameter $\Theta(h^{-1/3})\cdot \diam_2(P)$. In general, if the diameter of every hole is $o(1)\cdot \diam_2(P)$, we can improve upon the bound $g(h)\leq O(h^{1/2})$ in Theorem~\ref{thm:main}.

\begin{theorem} \label{thm:Delta}
If $P\in \C(h)$ and the diameter of every hole is at most $\Delta\cdot \diam_2(P)$, then $\varrho(P)\leq O(1+\min\{h^{3/4}\Delta,h^{1/2}\Delta^{1/2}\})$.
In particular for $\Delta=O(h^{-1/3})$, we have $\varrho(P)\leq O(h^{1/3})$.
\end{theorem}

However, if we further restrict the holes to be \emph{fat} convex polygons, we can show that $\varrho(P)= O(1)$ for all $h\in \NN$. 
In fact for every $s,t \in P$, the distortion $\geod(s,t)/|st|$ is also bounded by a constant. 

Informally, a convex body is \emph{fat} if its width is comparable with its diameter.
The \emph{width} of a convex body $C$ is the minimum width of a parallel slab enclosing~$C$. 
For every $\lambda \in (0,\lambda]$, a convex body $C$ is $\lambda$-\emph{fat} if the ratio of its width to its diameter is at least $\lambda$, that is, $\width(C)/\diam_2(C)\geq \lambda$; and $C$ is \emph{fat} if the inequality holds for a constant $\lambda\in (0,1]$. 
For instance, a disk is 1-fat, a $3 \times 4$ rectangle is $\frac35$-fat and a line segment is $0$-fat.
Let $\mathcal{F}_\lambda(h)$ be the family of polygonal domain $P=P_0\setminus \left(\bigcup_{i=1}^h P_i\right)$, where $P_0$ is a convex polygon and $P_1,\ldots , P_h$ are $\lambda$-fat convex polygons.

\begin{proposition} \label{prop:fat}
For $P\in \mathcal{F}_\lambda(h)$, where $h\in \NN$ and $\lambda\in (0,1]$, we have 
$\varrho(P)\leq O(\lambda^{-1})$. 
\end{proposition}

The special case when all holes are axis-aligned rectangles is also easy.

\begin{proposition} \label{prop:axisaligned}
Let $P\in \C(h)$, $h\in \NN$, such that all holes are axis-aligned rectangles. 
Then $\varrho(P)\leq O(1)$. 
\end{proposition}

\paragraph{Triangulations.} In this paper, we focus on the diameter distortion $\varrho(P)=\diam_g(P)/\diam_2(P)$ for polygons $P\in \C(h)$ with $h$ holes. 
Alternatively, we can compare the geodesic and Euclidean diameters in triangulations of $n$ points in the plane.
In a \emph{geometric graph} $G=(V,E)$, the vertices are distinct points in the plane, and the edges are straight-line segments between pairs of vertices. The \emph{Euclidean diameter} of $G$, 
$\diam_2(G)=\max_{u,v\in V} |uv|$ is the maximum distance between two vertices, and the \emph{geodesic diameter} $\diam_g(G)=\max_{u,v\in V} \dist(u,v)$, where $\dist(u,v)$ is the shortest path distance in $G$, i.e., the minimum Euclidean length of a $uv$-path in $G$.  With this notation, we define $\varrho(G)=\diam_g(G)/\diam_2(G)$, 

A Euclidean triangulation $T=(V,E)$ of a point set $V$ is a planar straight-line graph where all bounded faces are triangles, and their union is the convex hull $\conv(V)$. Let 
\begin{equation}\label{eq:f}
    f(n)=\sup_{G\in \mathcal{T}(n)} \varrho(G),
\end{equation}
where the supremum is taken over the set $\mathcal{T}(n)$ all $n$-vertex triangulations. 
Recall that $g(n)$ is the supremum of diameter distortions over polygons with $n$ convex holes; see~\eqref{eq:g}. We prove that $f(n)$ and $g(n)$ have the same growth rate.

\begin{theorem}\label{thm:tri}
    We have $g(n)=\Theta(f(n))$. 
\end{theorem}

\paragraph{Alternative problem formulation.}
The following version of the question studied here may be more attractive to the escape community~\cite{FW04,KubelL21}. 
Given $n$ pairwise disjoint convex obstacles in a convex polygon of unit diameter (\eg, a square),  what is the maximum length of a (shortest) escape route from any given point in the polygon to its boundary? According to Theorem~\ref{thm:main}, it is always $O(n^{1/2})$ and sometimes $\Omega(n^{1/3})$. 
\old{A discrete version is the escape problem in a $k \times k$ grid with $n$ point obstacles, where a path must follow grid edges. Note that in this version, a solution does not always exist. Both versions naturally extend to higher dimensions. 
} 

\paragraph{Related previous work.}
We are not aware of any previous work on the distortion between geodesic and Euclidean diameters. However, the geodesic distance in polygons with or without holes has been studied extensively from the algorithmic perspective; see~\cite{Mitchell17} for a comprehensive survey. We briefly review the state of the art. In a simple polygon $P$ with $n$ vertices, one can compute the geodesic distance between two given points in $O(n)$  time~\cite{LP84}, trade-offs are also available between time and workspace~\cite{Har15}. A shortest-path data structure can report the geodesic distance between any two query points in $O(\log n)$ time after $O(n)$ preprocessing time~\cite{GH89}. In $O(n)$ time, one can also compute the geodesic diameter~\cite{HS97} and radius~\cite{ABB16}.

For polygons with holes, more involved techniques are needed. Let $P$ be a polygon with $h$ holes, and a total of $n$ vertices. For any $s,t\in P$, one can compute $\geod(s,t)$ in 
$O(n+h\log h)$ time and $O(n)$ space~\cite{Wang23}, improving earlier bounds in~\cite{HS99,KMM97,Mitchell96,Wan21a}. A shortest-path data structure can report the geodesic distance between two query points in $O(\log n)$ query time using $O(n^{11})$ space; or in $O(h \log n)$ query time with $O(n+h^5)$ space~\cite{CM99}. 
The geodesic radius can be computed in $O(n^{11}\log n)$ time~\cite{BaeKO19,Wang18},
and the geodesic diameter in $O(n^{7.73})$ or $O(n^7 (\log n + h))$ time~\cite{BaeKO13}.
One can find a $(1+\eps)$-approximation of both the geodesic diameter and radius in $O((n/\eps^2 + n^2/\eps) \log n)$ time~\cite{BaeKO13,BaeKO19}. The geodesic diameter may be attained by a point pair $s,t\in P$, where both $s$ and $t$ lie in the interior or $P$; in which case it is known~\cite{BaeKO13} that there are at least five different geodesic paths between $s$ and $t$.

The diameter of an $n$-vertex triangulation with Euclidean lengths  can be computed in $\tilde{O}(n^{5/3})$ time~\cite{Cabello19,GawrychowskiKMS21}. For unweighted graphs in general, the diameter problem has been intensely studied in the fine-grained complexity community. For a graph with $n$ vertices and $m$ edges, breadth-first search (BFS) yields a 2-approximation in $O(m)$ time. Under the Strong Exponential Time Hypothesis (SETH), for any integer $k\geq 2$ and $\eps>0$, a $(2-\frac{1}{k}-\eps)$-approximation requires $mn^{1+1/(k-1)-o(1)}$ time~\cite{DalirrooyfardLW21}; see also~\cite{RodittyW13}.

\section{Convex Polygons with Convex Holes}
\label{sec:convex}

In this section, we prove Theorem~\ref{thm:main}. A lower bound construction is presented in Lemma~\ref{lem:lower}, and the upper bound is established in Lemma~\ref{lem:upper} below.

\paragraph{Lower Bound.}
The lower bound is based on the following construction. 
%
\begin{lemma}\label{lem:lower}
For every $h\in \NN$, there exists a polygonal domain $P\in \C(h)$ such that $g(P)\geq \Omega(h^{1/3})$.
\end{lemma}
\begin{proof}
We may assume, without loss of generality, that $h=k^3$ for some integer $k\geq 3$. We construct a polygon $P$ with $h$ holes, where the outer polygon $P_0$ is a regular $k$-gon of unit diameter, hence $\diam_2(P)=\diam_2(P_0)=1$. Let $Q_0,Q_1,\ldots , Q_{k^2}$ be a sequence of $k^2+1$ regular $k$-gons with a common center such that $Q_0=P_0$, and for every $i\in \{1,\ldots ,k^2\}$, $Q_i$ is inscribed in $Q_{i-1}$ such that the vertices of $Q_i$ are the midpoints of the edges of $Q_{i-1}$; see Fig.~\ref{fig:hexagons}. Enumerate the $k^3$ edges of $Q_1,\ldots, Q_{k^2}$ as $e_1,\ldots, e_{k^3}$. For every $j=1,\ldots , k^3$, we construct a hole as follows: Let $P_j$ be an $(|e|-2\eps)\times \frac{\eps}{2}$ rectangle with symmetry axis $e$ that contains $e$ with the exception of the $\eps$-neighborhoods of its endpoints. Then $P_1,\ldots ,P_{k^3}$ are pairwise disjoint. Finally, let $P=P_0\setminus \bigcup _{j=1}^{k^3} P_j$.

	\begin{figure}[htbp]
		\centering
		\includegraphics[width=\textwidth]{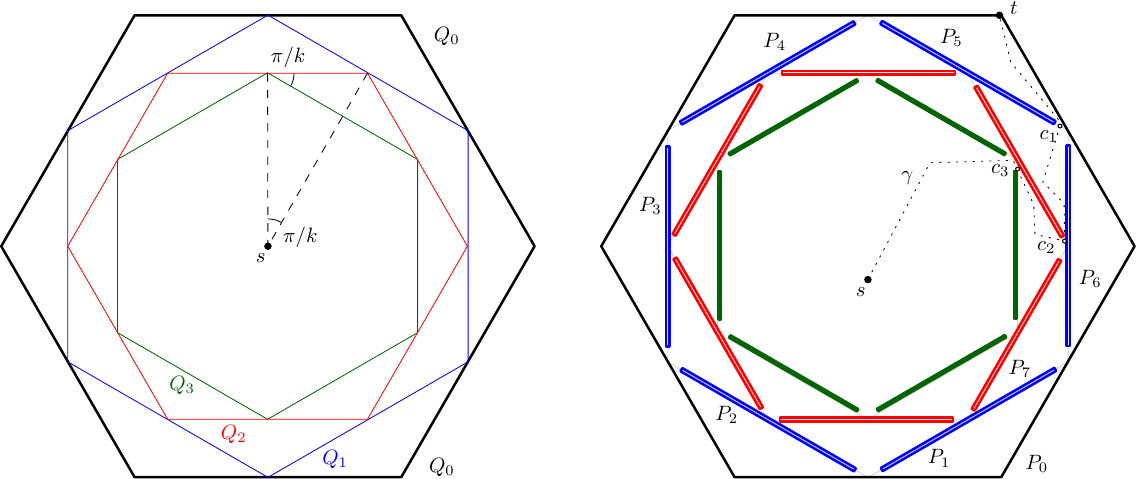}
		\caption{Left: hexagons $Q_0, \ldots, Q_3$ for $k=6$.   
         Right: The 18 holes corresponding to the edges of $Q_1,\ldots , Q_3$.}
		\label{fig:hexagons}
	\end{figure}

Assume, without loss of generality, that $e_i$ is an edge of $Q_i$ for $i\in \{0,1,\ldots , k^2\}$. As $P_0=Q_0$ is a regular $k$-gon of unit diameter, then $|e_0|\geq \Omega(1/k)$. Let us compare the edge lengths in two consecutive $k$-gons. Since $Q_{i+1}$ is inscribed in $Q_i$, we have 
\[|e_{i+1}|=|e_i|\cos \frac{\pi}{k} \geq |e_i| \left(1-\frac{\pi^2}{2k^2}\right) \]
using the Taylor estimate $\cos x\geq 1-x^2/2$.  Consequently, for every $i\in \{0,1,\ldots , k^2\}$, 
$$|e_i|
\geq |e_0|\cdot \left(1-\frac{\pi^2}{2k^2}\right)^{k^2} 
\geq |e_0|\cdot \Omega(1)
\geq \Omega\left(\frac{1}{k}\right).$$

It remains to show that $\diam_g(P)\geq \Omega(k)$. Let $s$ be the center of $P_0$ and $t$ an arbitrary vertex of $P_0$. Consider an $st$-path $\gamma$ in $P$, and for any two points $a,b$ along $\gamma$, let $\gamma(a,b)$ denote the subpath of $\gamma$ between $a$ and $b$. Let $c_i$ be the first point where $\gamma$ crosses the boundary of $Q_i$ for $i\in \{1,\ldots ,k^2\}$. By construction, $c_i$ must be in an $\eps$-neighborhood of a vertex of $Q_i$. Since the vertices of $Q_{i+1}$ are at the midpoints of the edges of $Q_i$, then $|\gamma(c_i,c_{i+1})| \geq \frac12\, |e_i|-2\eps \geq 
\Omega(|e_i|)\geq \Omega(1/k)$. Summation over $i=0,\ldots , k^2-1$ yields $|\gamma|\geq \sum_{i=0}^{k^2-1}|\gamma(c_i,c_{i+1})|\geq k^2\cdot \Omega(1/k)\geq \Omega(k)=\Omega(h^{1/3})$, as required. 
\end{proof}

\paragraph{Upper Bound.} 
Let $P\in \C(h)$ for some $h\in \NN$ and let $s\in P$. For every hole $P_i$, let $\ell_i$ and $r_i$ be points on the boundary of $P_i$ such that $\overrightarrow{s\ell_i}$ and $\overrightarrow{sr_i}$ are tangent to $P_i$, and $P_i$ lies on the left (respectively, right) side of the ray $\overrightarrow{s\ell_i}$ (respectively, $\overrightarrow{sr_i}$).

	\begin{figure}[htbp]
		\centering
		\includegraphics[width=\textwidth]{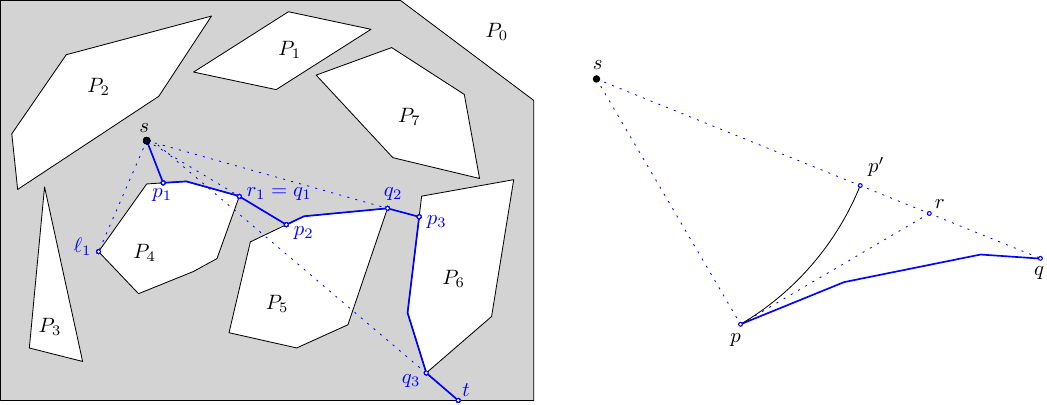}
		\caption{Left: A polygon $P\in \C(7)$ with 7 convex holes, a point $s\in P$,
       and a path $\greedy_P(s,\vec{u})$ from $s$ to a point $t$ on the outer boundary of $P$.
       Right: A boundary arc $\widehat{pq}$, where $|\widehat{pq}|\leq |pr|+|rq|$. }
		\label{fig:example}
	\end{figure}

We construct a path from $s$ to some point on the outer boundary of $P$ by the following recursive algorithm; refer to Fig.~\ref{fig:example}~(left). For a unit vector $\vec{u}\in \mathbb{S}^1$, we 
construct a path $\greedy_P(s,\vec{u})$ as follows. Start from $s$ along a ray emanating from $s$ in direction $\vec{u}$ until it reaches the boundary of $P$ at some point $p$. While $p\notin \partial P_0$ do: Assume that $p\in \partial P_i$ for some $1\leq i\leq h$. Extend the path along $\partial P_i$ to the point $\ell_i$ or $r_i$ so that the distance from $s$ monotonically increases along the path; and then continue along the ray $\overrightarrow{s\ell_i}$ or $\overrightarrow{sr_i}$ until it reaches the boundary of $P$ again. When $p\in \partial P_0$, the path $\greedy_P(s,\vec{u})$ terminates at $p$.

\begin{lemma}\label{lem:upper}
For every $P\in \C(h)$, every $s\in P$ and every $\vec{u}\in \mathbb{S}^1$, we have $|\greedy_P(s,\vec{u})|\leq O(h^{1/2})\cdot \diam_2(P)$, and this bound is the best possible.
\end{lemma}
\begin{proof}
Let $P$ be a polygonal domain with a convex outer polygon $P_0$ and $h$ convex holes.
We may assume, without loss of generality, that $\diam_2(P)=1$. For a point $s\in P$ and a unit vector $\vec{u}$, consider the path $\greedy_P(s,\vec{u})$. By construction, the distance from $s$ monotonically increases along this path, and so the path has no self-intersections. It is composed of \emph{radial segments} that lie along rays emanating from $s$, and \emph{boundary arcs} that lie on the boundaries of holes. By monotonicity, the total length of all radial segments is at most $\diam_2(P)$. Since every boundary arc ends at a point of tangency $\ell_i$ or $r_i$, for some $i\in \{1,\ldots, h\}$, the path $\greedy_P(s,\vec{u})$ contains at most two boundary arcs along each hole (the bound of $2$ can be attained; see Fig.~\ref{fig:twice}(left) for an example), Therefore, the number of boundary arcs is at most $2h$.
. 
Let $\A$ denote the set of all boundary arcs along $\greedy_P(s,\vec{u})$; then $|\A|\leq 2h$.

Along each boundary arc $\widehat{pq}\in \A$, from $p$ to $q$, the distance from $s$ increases by $\Delta_{pq}=|sq|-|sp|$. By monotonicity, we have $\sum_{\widehat{pq}\in \A} \Delta_{pq}\leq \diam_2(P)$.
We now give an upper bound for the length of $\widehat{pq}$. Let $p'$ be a point in $sq$ such that $|sp|=|sp'|$, and let $r$ be the intersection of $sq$ with a line orthogonal to $sp$ passing through $p$; see Fig.~\ref{fig:example}~(right).
Note that $|sp|<|sr|$. Since the distance from $s$ monotonically increases along the arc $\widehat{pq}$, then $q$ is in the closed halfplane bounded by $pr$ that does not contain $s$. Combined with $|sp|<|sr|$, this implies that $r$ lies between $p'$ and $q$ on the line $sq$, 
consequently $|p'r| <|p'q|=\Delta_{pq}$ and $|rq| <|p'q|=\Delta_{pq}$.
By the triangle inequality and the Pythagorean theorem, these estimates give an upper bound
\begin{align*}
|\widehat{pq}|
&\leq |pr|+|rq|
= \sqrt{|sr|^2-|sp|^2} + |rq|
=  \sqrt{(|sp'|+|p'r|)^2-|sp|^2 } +|rq| \\    
&\leq  \sqrt{(|sp|+\Delta_{pq})^2-|sp|^2 } +\Delta_{pq}
\leq  O\left(\sqrt{|sp| \Delta_{pq}} +\Delta_{pq} \right) \\
&\leq  O\left(\sqrt{\diam_2(P)\cdot \Delta_{pq}} +\Delta_{pq} \right).
\end{align*}
Summation over all boundary arcs, using Jensen's inequality, yields
\begin{align*}
\sum_{\widehat{pq}\in \A} |\widehat{pq}| 
&\leq \sum_{\widehat{pq}\in \A}  O\left(\sqrt{\diam_2(P)\cdot \Delta_{pq}}+  \Delta_{pq}\right)\\
&\leq \sqrt{\diam_2(P)} \cdot O\left(\sum_{\widehat{pq}\in \A} \sqrt{\Delta_{pq}}\right)
    + O\left(\sum_{\widehat{pq}\in \A}  \Delta_{pq}\right)\\
&\leq \sqrt{\diam_2(P)}\cdot O\left( |\A| \cdot 
     \sqrt{\frac{1}{|\A|}\sum_{\widehat{pq}\in \A} \Delta_{pq}}\right)  + O(\diam_2(P))\\
&\leq \sqrt{\diam_2(P)}\cdot O\left(\sqrt{|\A|\cdot \diam_2(P)}  
\right)  + O(\diam_2(P))\\
&\leq  O\left(\sqrt{|\A|}\right)\cdot \diam_2(P) 
\leq  O\left(\sqrt{h}\right)\cdot \diam_2(P) ,
\end{align*}
as claimed.

	\begin{figure}[htbp]
		\centering
		\includegraphics[width=\textwidth]{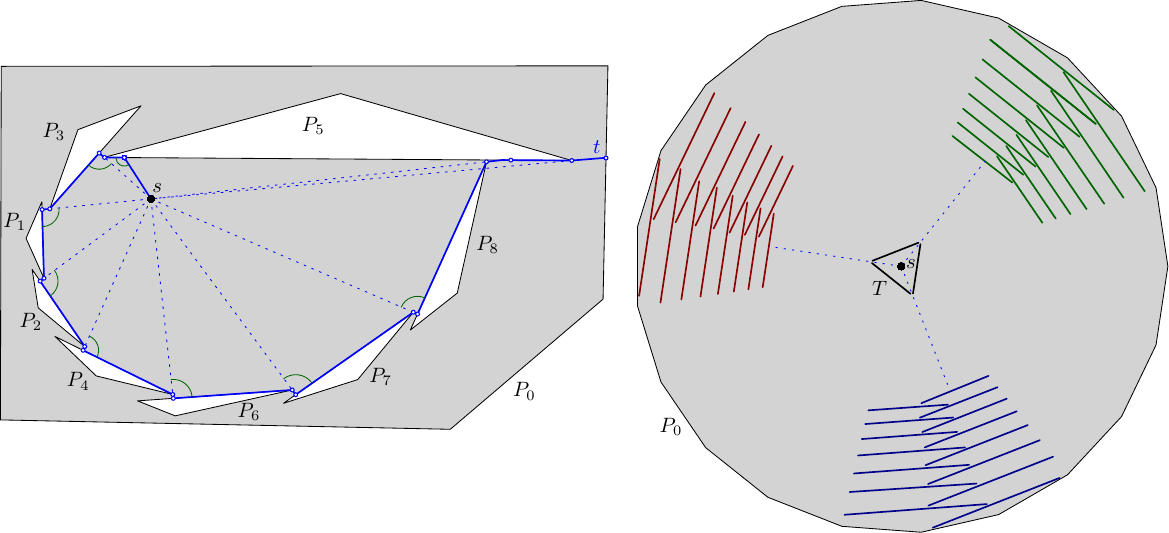}
		\caption{Left: An example, where $\greedy_P(s,\vec{u})$ visits $P_5$ twice.
        Right: An example, where the bound $|\greedy_P(s,\vec{u})|\leq O(h^{1/2})\cdot \diam_2(P)$ is attained for all $\vec{u}\in \mathbb{S}^1$.}
		\label{fig:twice}
	\end{figure}

We now show that the bound  $|\greedy_P(s,\vec{u})|\leq O(h^{1/2})\cdot \diam_2(P)$ is the best possible. For every $h\in \NN$, we construct a polygon $P\in \C(h)$ and a point $s$ such that for every $\vec{u}\in \mathbb{S}^1$, we have $|\greedy_P(s,\vec{u})|\geq \Omega(h^{1/2})$. Without loss of generality, we may assume that 
$\diam_2(P)=1$, $h=3(k^2+1)$, and $k$ is a multiple of 3. 

We start with the construction described in the proof of Lemma~\ref{lem:lower} with $k^3$ rectangular holes in a regular $k$-gon $P_0$, where $s$ is the center of $P_0$. We modify the construction in three steps; see Fig.~\ref{fig:twice}(right): 
(1) Let $T=(a,b,c)$ be a small equilateral triangle centered at $s$, such that the rays $\overrightarrow{sa}$, $\overrightarrow{sb}$, and $\overrightarrow{sc}$ intersect the midpoints of three edges of the  innermost $k$-gon (this is possible because $k$ is a multiple of 3).
Then construct three rectangular holes around the edges of $T$; to obtain a total of $k^3+3$ holes. (2) Rotate each rectangular hole $P_j$ around its own center by a sufficiently small angle as follows: The holes corresponding to the $i$-th $k$-gons from $s$ are rotated counterclockwise if $i$ is odd, and clockwise if $i$ is even. 
As a result, when the path $\greedy_P(s,\vec{u}$ successively reaches holes in an $\eps$-neighborhood of their center, it alternately turns right or left. (3) For every $\vec{u}\in \mathbb{S}^1$, the path $\greedy_P(s,\vec{u})$ exits the triangle $T$ at a small neighborhood of a corner of $T$. From each corner of $T$, $\greedy_P(s,\vec{u})$ continues to the outer boundary along a sequence of $k^2$ holes. 
We delete all holes that $\greedy_P(s,\vec{u})$ does not touch for any $\vec{u}\in \mathbb{S}^1$, thus we retain $h=3k^2+3$ holes. For every $\vec{u}\in \mathbb{S}^1$, we have $|\greedy_P(s,\vec{u})|\geq \Omega(k)$ according to the analysis in Lemma~\ref{lem:lower},
hence $|\greedy_P(s,\vec{u})|\geq \Omega(h^{1/2})$, as required. 
\end{proof}

\begin{corollary}
For every $h\in \NN$ and every polygon $P\in \C(h)$, we have 
$\diam_g(P)\leq O(h^{1/2})\cdot \diam_2(P)$.
\end{corollary}
\begin{proof}
Let $P\in \C(h)$ and $s_1,s_2\in P$. By Lemma~\ref{lem:upper}, there exist points $t_1,t_2\in \partial P_0$ such that $\geod(s_1,t_1)\leq O(h^{1/2})\cdot \diam_2(P)$ and $\geod(s_2,t_2)\leq O(h^{1/2})\cdot \diam_2(P)$. There is a path between $t_1$ and $t_2$ along the perimeter of $P_0$. It is well known~\cite{SA00,ConvexFigures} that $|\partial P_0|\leq \pi\cdot \diam_2(P_0)$ for every convex body $P_0$, hence $\geod(t_1,t_2)\leq O(\diam_2(P))$. The concatenation of these three paths yields a path in $P$ connecting $s_1$ and $s_2$, of length $\geod(s_1,s_2)\leq O(h^{1/2})\cdot \diam_2(P)$.
\end{proof}

\section{Improved Upper Bound for Holes of Bounded Diameter}  \label{sec:improved}

In this section, we prove Theorem~\ref{thm:Delta}. Similarly to the proof of Theorem~\ref{thm:main}, it is enough to bound the geodesic distance from an arbitrary point in $P$ to the outer boundary. We give three such bounds in Lemmas~\ref{lem:unit}, \ref{lem:Delta} and~\ref{lem:Delta2} that address different ranges of $h$: Lemmas~\ref{lem:unit} covers $\Delta<O(h^{-1})$ and Lemma~\ref{lem:Delta2} the range $\Delta\geq h^{-1}$; while Lemma~\ref{lem:Delta} holds for all values of $\Delta$. 

\begin{lemma}\label{lem:unit}
Let $P\in \C(h)$ such that $\diam_2(P_i) \leq \Delta\cdot \diam_2(P)$ for every hole $P_i$. If $\Delta\leq O(h^{-1})$, then from every point $s\in P$, there exists a path of length $O(\diam_2(P))$ in $P$ to the outer boundary $\partial P_0$. 
\end{lemma}
\begin{proof}
Let $s\in P$ and $t\in \partial P_0$. Construct an $st$-path $\gamma$ as follows: Start with the straight line segment $st$, and whenever $st$ intersects the interior of a hole $P_i$, the segment $st\cap P_i$ is replaced by a shortest path along $\partial P_i$ between the endpoints of $st\cap P_i$. Since $|\partial P_i|\leq \pi\cdot \diam_2(P_i)$ for every convex hole $P_i$~\cite{SA00,ConvexFigures}, then $|\gamma|\leq |st|+\sum_{i=1}^h |\partial P_i| \leq \diam_2(P)+ \sum_{i=1}^h O(\diam_2(P_i)) \leq \O(1+h\Delta)\cdot \diam_2(P) \leq O(\diam_2(P))$, as claimed.
\end{proof}

\begin{lemma}\label{lem:Delta}
Let $P\in \C(h)$ such that $\diam_2(P_i) \leq \Delta\cdot \diam_2(P)$ for every hole $P_i$. Then from every point $s\in P$, there exists a path of length $O(1+h^{3/4} \Delta) \cdot \diam_2(P)$ in $P$  to the outer boundary $\partial P_0$. 
\end{lemma}
\begin{proof}
Assume without loss of generality that $\diam_2(P)=1$ and $s$ is the origin. Let $\ell\in \NN$ be a parameter to be specified later. For $i\in \{-\ell, -\ell+1,\ldots , \ell\}$, 
let $H_i: y=i\cdot \Delta$ be a horizontal line, and $V_i: x=i\cdot \Delta$ be a vertical line. Since any two consecutive horizontal (respectively, vertical) lines are distance $\Delta$ apart, and the diameter of each hole is at most $\Delta$, then the interior of each hole intersects at most one horizontal and at most one vertical line. By the pigeonhole principle, there are integers $a,b,c,d\in \{1,\ldots , \ell\}$ such that $H_{-a}$, $H_{b}$, $V_{-c}$, and $V_{d}$ each intersect the interior of at most $h/\ell$ holes; see Fig.~\ref{fig:shorts}. 

\begin{figure}[htbp]
		\centering 
\includegraphics[width=\textwidth]{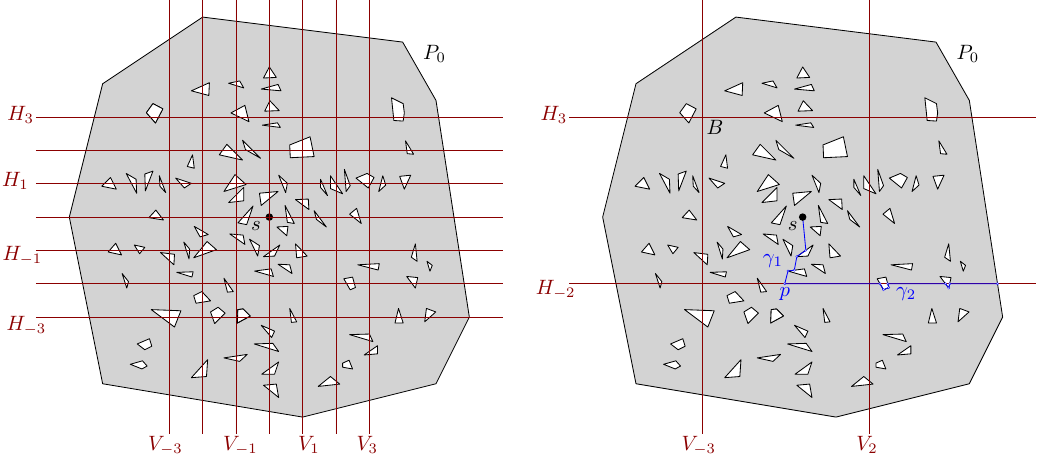}
		\caption{Illustration for $\ell=3$. Left: All lines $H_i$ and $V_i$. Right: Box $B$ and a path $\gamma$ from $s$ to the outer boundary.}
		\label{fig:shorts}
	\end{figure}

Let $B$ be the axis-aligned rectangle bounded by the lines $H_{-a}$, $H_{b}$, $V_{-c}$, and $V_{d}$. Due to the spacing of the lines, we have $\diam_2(B)\leq 2\cdot \sqrt{2}\cdot\ell\Delta=O(\ell\Delta)$.

We can construct a path from $s$ to $\partial P_0$ as a concatenation of two paths $\gamma=\gamma_1 \oplus \gamma_2$. Let $\gamma_1$ be the initial part of $\greedy_P(s,\vec{u})$ from $s$ for an arbitrary direction $\vec{u}$ until it reaches the boundary of $B\cap P_0$ at some point $p$. If $p\in \partial P_0$, then $\gamma_2=(p)$ is a trivial one-point path. Otherwise, $p$ lies on a line $L\in \{H_{-a},H_b,V_{-c},V_d\}$ that intersects the interior of at most $h/\ell$ holes.
Let $\gamma_2$ follow $L$ from $p$ to the boundary of $P_0$ so that when it encounters a hole $P_i$, it makes a detour along $\partial P_i$, where each detour is the shortest path between two points along $\partial P_i$. 

It remains to analyze the length of $\gamma$. By Lemma~\ref{lem:upper}, we have $|\gamma_1|\leq O(\sqrt{h})\cdot \diam_2(B)\leq O(h^{1/2}\ell\Delta)$.
The path $\gamma_2$ has edges along the line $L$ and along the boundaries of holes whose interior intersect $L$. The total length of all edges along $L$ is at most $\diam_2(P)=1$.
It is well known that $\per(C)\leq \pi\cdot \diam_2(C)$ for every convex body~\cite{SA00,ConvexFigures}, and so the length of each detour is $O(\diam_2(P_i))\leq O(\Delta)$, and the total length of $O(h/\ell)$ detours is $O(h\Delta/\ell)$. Consequently,
\begin{equation}\label{eq:balance}
|\gamma|\leq O(h^{1/2}\ell\Delta+h\Delta/\ell+1).
\end{equation}
Finally, we set $\ell=\lceil h^{1/4}\rceil$ to balance the first two terms in \eqref{eq:balance},  and obtain $|\gamma|\leq O(h^{3/4}\Delta+1)$, as claimed.
\end{proof}

When all holes are line segments, we construct a monotone path from $s$ to the outer boundary. A polygonal path $\gamma=(p_0,p_1,\ldots , p_m)$ is \emph{$\vec{u}$-monotone} for a unit vector $\vec{u}\in \mathbb{S}^1$ if $\vec{u}\cdot \overrightarrow{v_{i-1} v_i} \geq 0$ for all $i\in \{1,\ldots, m\}$; and $\gamma$ is \emph{monotone} if it is $\vec{u}$-monotone for some $\vec{u}\in \mathbb{S}^1$.

\begin{lemma}\label{lem:segment}
Let $P\in \C(h)$ such that every hole is a line segment of length at most $\Delta\cdot \diam_2(P)$. If $\Delta\geq h^{-1}$, then from every point $s\in P$, there exists a monotone path of length $O(h^{1/2}\Delta^{1/2}) \cdot \diam_2(P)$ in $P$ to the outer boundary $\partial P_0$.
\end{lemma}
\begin{proof}
We may assume, without loss of generality, that $\diam_2(P)=1$.
Denote the line segments by $a_ib_i$, for $i=1,\ldots, h$, such that $x(a_i)\leq x(b_i)$. Let $\ell=\lceil h^{1/2}\Delta^{1/2}\rceil$, and note that $\ell = \Theta(h^{1/2}\Delta^{1/2})$ when $\Delta\geq h^{-1}$. Partition the right halfplane (i.e., right of the $y$-axis) into $\ell$ wedges with aperture $\pi/\ell$ and apex at the origin, denoted $W_1,\ldots , W_\ell$. For each wedge $W_i$, let $\vec{w}_i\in \mathbb{S}$ be the direction vector of its axis of symmetry.

Partition the $h$ segments as follows: For $j=1,\ldots , \ell$, let $\mathcal{H}_j$ be the set 
of segments $a_ib_i$ such that $\overrightarrow{a_ib_i}$ is in $W_j$. 
Finally, let $\mathcal{H}_{j^*}$ be a set with minimal cardinality, that is, $|\mathcal{H}_{j^*}|\leq h/\ell = O(h^{1/2}/\Delta^{1/2})$.
Let $\vec{v}=\vec{w}_{j^*}^\perp$.
We construct a $\vec{v}$-monotone path $\gamma$ from $s$ to the outer boundary $\partial P_0$ as follows. Start in direction $\vec{v}$ until the path reaches a hole $a_ib_i$ at some point $p$. While $p\notin \partial P_0$, continue along $a_ib_i$ to one of the endpoints: 
to $a_i$ if $\vec{v} \cdot \overrightarrow{a_ib_i}\geq 0$, and to $b_i$ otherwise;
then continue in direction $\vec{v}$. By monotonicity, $\gamma$ visits every edge at most once. 

It remains to analyze the length of $\gamma$. We distinguish between two types of edges: let $E_1$ be the set of edges of $\gamma$ contained in $\mathcal{H}_{j^*}$, and $E_2$ be the set of all other edges of $\gamma$. The total length of edges in $E_1$ is at most the total length of all segments in $\mathcal{H}_{j^*}$, that is, 
\begin{align*}
\sum_{e\in E_1}|e|
&\leq |\mathcal{H}_{j^*}| \cdot \Delta
  \leq O(h^{1/2}/\Delta^{1/2})\cdot \Delta 
  = O(h^{1/2}\Delta^{1/2}).
\end{align*}

Every edge $e\in E_2$ makes an angle at least $\pi/(2\ell)$ with vector $\vec{v}$. Let $\mathrm{proj}(e)$ denote the orthogonal projection of $e$ to a line of direction $\vec{v}$. Then $|\mathrm{proj}(e)| \geq  |e| \sin (\pi/(2\ell))$.
By monotonicity, the projections of distinct edges have disjoint interiors. 
Consequently, $\sum_{e\in E_2} |\mathrm{proj}(e)| \leq \diam_2(P)=1$. 
This yields 
\begin{align*}
\sum_{e\in E_2} |e| 
&\leq \sum_{e\in E_2} \frac{|\mathrm{proj}(e)|}{\sin(\pi/(2\ell))}
 = \frac{1}{\sin(\pi/(2\ell))} \sum_{e\in E_2} |\mathrm{proj}(e)|\\
&= O(\ell) =O(h^{1/2} \Delta^{1/2}).
\end{align*}
Overall, $|\gamma| = \sum_{e\in E_1}|e|+  \sum_{e\in E_2}|e| = O(h^{1/2}\Delta^{1/2})$, as claimed.
\end{proof}

	\begin{figure}[htbp]
		\centering
		\includegraphics[width=0.6\textwidth]{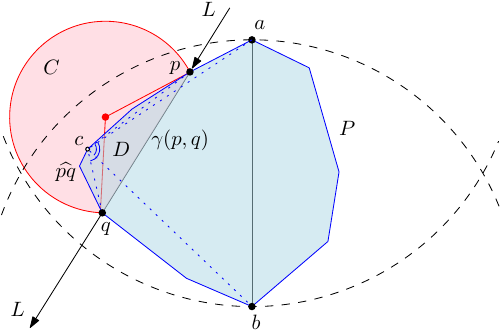}
		\caption{Line $L$ traverses a convex polygon $P$, but does not cross the segment $ab$.}
		\label{fig:dilation}
	\end{figure}
 
For extending Lemma~\ref{lem:segment} to arbitrary convex holes, we need the following technical lemma; refer to Fig.~\ref{fig:dilation}. 

\begin{lemma}\label{lem:tech}
Let $P$ be a convex polygon with a diametral pair $a,b\in \partial P$, where $|ab|=\diam_2(P)$. Suppose that a line $L$ intersects the interior of $P$, but does not cross the line segment $ab$. Let $p,q\in \partial P$ such that $pq=L\cap P$, and points $a$, $p$, $q$, and $b$ appear in this counterclockwise order in $\partial P$; and let $\widehat{pq}$ be the counterclockwise $pq$-arc of $\partial P$. Then $|\widehat{pq}|\leq \frac{4\pi\sqrt{3}}{9}\, |pq| < 2.42|pq|$.
\end{lemma}
\begin{proof}
Since $|ab| = \diam_2(P)$, then $P$ lies in the intersection of two disks of radius
$|ab|$ centered at $a$ and $b$, respectively. It follows that for any point $c\in P$, we
have $\angle a c b \geq \pi/3$. 
For any point $c\in \widehat{pq}$, we have $\angle pcq \geq \angle a c b$ by convexity,
whence $\angle pcq \geq \angle a c b \geq \pi/3$.
The locus of points $c$ where $(p,c,q)$ is a counterclockwise triple and $\angle pcq \geq \pi/3$ is the convex compact set $C$ bounded by $pq$ and a circular arc of radius $\frac{\sqrt{3}}{3}\, |pq|$. The length of the circular arc is $\frac{2\pi/3}{2\pi}=\frac23$ times the circumference of the circle, or $\frac23 \cdot 2\pi\cdot \frac{\sqrt{3}}{3}\, |pq| = \frac{4\pi\sqrt{3}}{9}\, |pq| < 2.42|pq|$. 

Let $D$ be the convex compact set bounded by $pq$ and $\widehat{pq}$. It is well known that for any two convex compact sets $C,D\subset \RR^2$, $C\subset D$ implies $\per(C)\leq \per(D)$ (e.g., by Crofton's formula). Consequently, $|\widehat{pq}|< 2.42\, |pq|$.
\end{proof}

\begin{lemma}\label{lem:Delta2}
Let $P\in \C(h)$ such that $\diam_2(P_i) \leq \Delta\cdot \diam_2(P)$ for every hole $P_i$. If $\Delta\geq h^{-1}$, then from every point $s\in P$, there exists a path of length $O(h^{1/2}\Delta^{1/2}) \cdot \diam_2(P)$ in $P$ to the outer boundary $\partial P_0$. 
\end{lemma}
\begin{proof}
We may assume without loss of generality that $\diam_2(P)=1$.
For each hole $P_i$, $i\in \{1,\ldots, h\}$, choose a diametral point pair $a_i,b_i\in \partial P_i$ with $\diam_2(P_i)=|a_ib_i|$, and $x(a_i)\leq x(b_i)$.
Note that $P_i$ is a convex hole, and so the segment $a_ib_i$ is disjoint from the interior of $P$. By Lemma~\ref{lem:segment}, there exists a monotone path $\gamma$ of length $|\gamma|\leq O(h^{1/2}\Delta^{1/2})$ from $s$ to the outer boundary $\partial P_0$, that lies in $P_0$, and does not cross any of the segments $a_ib_i$.
The edges of $\gamma$ alternate between edges along segments $a_ib_i$, which lie in distinct holes, and edges between two distinct segments that are all parallel to a common direction vector $\vec{w}$. However, $\gamma$ may intersect some of the holes, so it does not necessarily lie in $P$.

We modify $\gamma$ as follows. For every maximal subpath $\gamma(p,q)$ whose interior lies in the interior of a hole $P_i$, we replace $\gamma(p,q)$ with a path $\widehat{pq}$ along $\partial P_i$. 
Note that the segments $pq$ and $a_ib_i$ do not cross: Indeed, suppose they cross. Then $\gamma(p,q)$ crosses $a_ib_i$ in the interior of $P_i$. However, when $\gamma$ reaches $a_ib_i$, it follows it to one of its endpoint, so $q\in \{a_i,b_i\}$, which is a contradiction. 

By Lemma~\ref{lem:tech}, $|\widehat{pq}|\leq 2.42\, |\gamma(p,q)|$. Consequently, all detours jointly increase the length of $\gamma$ by a factor of at most 2.42, and the resulting path has length $O(h^{1/2}\Delta^{1/2})$.
\end{proof}

\section{Polygons with Fat or Axis-Aligned Convex Holes}
\label{sec:fat}

In this section, we show that in a polygonal domain $P$ with fat convex holes, the distortion $\geod(s,t)/|st|$ is bounded by a constant for all $s,t\in P$, and prove Proposition~\ref{prop:fat}. 

Let $C$ be a convex body in the plane, and let $P=\RR^2\setminus C$ be its complement. For every pair of points $s,t\in \partial C=\partial P$, we compare the Euclidean distance $|st|$ with the geodesic distance $\geod(s,t)$ with respect to $P$, which is the shortest $st$-path along the boundary of $\partial C$. The \emph{geometric dilation} of $C$ is $\delta(C)=\sup_{s,t\in \partial C} \frac{\geod(s,t)}{|st|}$.

\begin{lemma} \label{lem:dilation}
  Let $C$ be a $\lambda$-fat convex body, $\lambda\in (0,1]$. Then $\delta(C) \leq \min\{\pi \lambda^{-1}, 2(\lambda^{-1} +1)\} = O(\lambda^{-1})$.
\end{lemma}
\begin{proof}
It is known~\cite[Lemma~11]{EGK07} that $\delta(C) = \frac{|\partial C|}{2h}$, where $h=h(C)$
  is the \emph{minimum halving distance} of $C$ (\ie, the minimum distance between two points on $C$ that divide the length of $C$ in two equal parts).
  It is also known~\cite[Thm.~8]{DEG+07} that $h \geq \width(C)/2$.
  Putting these together one deduces that $\delta(C) \leq \frac{|\partial C|}{\width(C)}$.
  The isoperimetric inequality $|\partial C| \leq \diam_2(C) \pi$ and the obvious inequality $|\partial C| \leq 2\diam_2(C) + 2\width(C)$ lead to the following
  dilation bounds: $\delta(C) \leq \pi\, \frac{\diam_2(C)}{\width(C)}$ and $\delta(C) \leq 2\left(\frac{\diam_2(C)}{\width(C)}+1\right)$; see also~\cite{DEG+07,SA00}. Since $C$ is $\lambda$-fat, direct substitution yields the two bounds in the lemma;
the latter bound is better for small~$\lambda$.
\end{proof}

\begin{corollary}\label{cor:dilation}
Let $P=P_0\setminus \left(\bigcup_{i=1}^h P_i\right)$ be a polygonal domain, where $P_0$ is a convex polygon and $P_1,\ldots, P_h$ are $\lambda$-fat convex polygons for some $\lambda\in (0,1]$. Then for every $s,t\in P$, we have 
$\geod(s,t)\leq O(\lambda^{-1} |st|)$.
\end{corollary}
\begin{proof}
If the line segment $st$ is contained in $P$, then $\geod(s,t)=|st|$, and the proof is complete. Otherwise, segment $st$ is the concatenation of line segments contained in $P$ and line segments $p_iq_i\subset P_i$ with $p_i,q_i\in \partial P_i$, for some indices $i\in \{1,\ldots , h\}$. By replacing each segment $p_iq_i$ with the shortest path on the boundary of the hole $P_i$, we obtain an $st$-path $\gamma$ in $P$. Since each hole is $\lambda$-fat, we replaced each line segment $p_iq_i$ with a path of length $O(|p_iq_i|/\lambda)$ by Lemma~\ref{lem:dilation}. Overall, we have $|\gamma|\leq O(|st|/\lambda)$, as required.
\end{proof}

\begin{corollary}
If $P=P_0\setminus \left(\bigcup_{i=1}^h P_i\right)$ be a polygonal domain, where $P_0$ is a convex polygon and $P_1,\ldots, P_h$ are $\lambda$-fat convex polygons for some $\lambda\in (0,1]$, then 
$\diam_g(P)\leq O(\lambda^{-1} \diam_2(P))$, hence $\varrho(P)\leq O(\lambda^{-1})$.
\end{corollary}

\begin{proposition} \label{prop:axis-aligned}
Let $P\in \C(h)$, $h\in \NN$, such that every hole is an axis-aligned rectangle. Then from every point $s\in P$, there exists a path of length at most $\diam_2(P)$ in $P$ to the outer boundary $\partial P_0$. 
\end{proposition}
\begin{proof}
Let $B=[0,a] \times [0,b]$ be a minimal axis-parallel bounding box containing $P$.
We may assume without loss of generality that $x(s) \geq a/2$, $y(s) \geq b/2$,
and $b \leq a$. We construct a staircase path $\gamma$ as follows. Start from $s$ in horizontal direction $\vec{d}_1=(1,0)$ until the path reaches  the boundary $\partial P$ at some point $p$. While $p\notin \partial P_0$, make a $90^\circ$ turn from $\vec{d}_1=(1,0)$ to $\vec{d}_2=(0,1)$ or vice versa, and continue. We have $|\gamma| \leq \frac{a+b}{2} \leq a \leq \diam_2(P)$,
as claimed.
\end{proof}

\section{Polygons with Holes versus Triangulations}
\label{sec:tri}

In section we prove Theorem~\ref{thm:tri}: the proof is the combination of Lemmas~\ref{lem:tri} and~\ref{lem:delaunay} below. 

\begin{lemma}\label{lem:tri}
For every triangulation $T\in \mathcal{T}(n)$, there exists a polygonal domain $P\in \C(h)$ with $h=\Theta(n)$ holes such that $\varrho(P)=\Theta(\varrho(T))$.
\end{lemma}
\begin{proof}
Given a triangulation $T\in \mathcal{T}(n)$, we construct a polygonal domain $P$ as follows. Let the outer polygon $P_0$ be the convex hull of $T$, and in the interior of each triangle $t_i$, create a triangular hole $P_i$ such that $\partial P_i$ is homothetic to $t_i$ and 
lies in the $\eps$-neighborhood of $t_i$ for some small $\eps>0$. A triangulation with $n$ vertices has at most $2n-5= O(n)$ triangular faces, and so $P\in \C(h)$ with $h=O(n)$ holes.

We claim that $\varrho(P)=\Theta(\varrho(T))$. Note that $\diam_2(P)=\diam_2(T)$ by construction, and so it is enough to prove that $\diam_g(P)=\Theta(\diam_g(T))$. Let $s,t\in P$ be a diametral pair, where $\geod_P(s,t)=\diam_g(P)$. By construction, $s$ and $t$ lie in the $\eps$-neighborhood of some edges $e_s,e_t\in E(T)$, respectively. Let $\gamma$ be a shortest $st$-path in $P$, and suppose that it passes through the $\eps$-neighborhoods of the vertices $(v_1,\ldots, v_k)$. Note that all edges and vertices of $T$ are contained in $P$. In particular, the shortest $v_1v_k$-path in $T$ is also contained in $P$, and so $\geod_P(v_1,v_k)\leq \geod_T(v_1,v_k)$. The geodesic distance between $s$ and the $\eps$-neighborhood of $v_1$ (respectively, $t$ and the $\eps$-neighborhood of $v_k$) is at most $\diam_2(T)\leq \diam_g(T)$. Overall, we obtain 
\[
    \geod_P(s,t)
   \leq \eps+|e_s|+\geod_T(v_1,v_k)+|e_t|+\eps
    \leq 3\, \diam_g(T)+2\eps 
    \leq O(\diam_g(T))
\]
if $\eps>0$ is sufficiently small, and so $\diam_g(P)\leq O(\diam_g(T))$.

Conversely, let $u,v\in V(T)$ be a diametral pair of vertices, i.e., $\geod_T(u,v)=\diam_g(T)$. 
Let $\gamma$ be a shortest $uv$-path in $P$. Clearly, we have $|\gamma|\leq \diam_g(P)$. 
Suppose that $\gamma$ intersects the $\eps$-neighborhoods of vertices $(v_1,\ldots , v_k)$, where $s=v_1$ and $t=v_k$. For $i=1,\ldots ,k-1$, denote by $\gamma_i$ the subpath of $\gamma$ between the $\eps$-neighborhoods of $v_i$ and $v_{i+1}$. Now $(v_1,\ldots , v_k)$ is an $st$-path in $T$ of length 
\[ \sum_{i=1}^{k-1} |v_iv_{i+1}| 
\leq \sum_{i=1}^{k-1}(|\gamma_i|+2\eps) 
< |\gamma| +  2k\eps
\leq O(\diam_g(P)+n\eps)
\leq O(\diam_g(P)), \]
if $\eps>0$ is sufficiently small. 
This implies $\diam_g(T)=\geod_T(u,v)\leq O(\diam_g(P))$,
and finally that $\varrho(P)=\Theta(\varrho(T))$, as required. 
\end{proof}

Every planar straight-line graph $G=(V,E)$ can be augmented to a triangulation $T=(V,E')$, with $E\subseteq E'$. One of the notable triangulations is the \emph{Constrained Delaunay Triangulation}, for short, $\mathrm{CDT}(G)$. 
Bose and Keil~\cite{BK06} proved that $\mathrm{CDT}(G)$ has bounded stretch for so-called \emph{visibility} edges. Specifically, if $u,v\in V$ and $uv$ does not cross any edge of $G$, then $\mathrm{CDT}(G)$ contains a $uv$-path of length $O(|uv|)$. The constant factor was later improved by Bose~\etal~\cite{BCR18,BoseFRV19}.

\begin{lemma}\label{lem:delaunay}
For every polygonal domain $P\in \C(h)$, there exists a triangulation $T\in \mathcal{T}(n)$ with $n=\Theta(h)$ vertices such that $\varrho(T)=\Theta(\varrho(P))$.
\end{lemma}
\begin{proof}
Assume that $P=P_0\setminus \bigcup_{i=1}^h P_i$. For all $j=1,\ldots , h$, let 
$a_i,b_i\in \partial P_i$ be a diametral pair, that is, $|a_ib_i|=\diam_2(P_i)$. 
The line segments $\{a_ib_i: i=1,\ldots , h\}$, together with the four vertices of 
a minimum axis-aligned bounding box of $P$, form a planar straight-line graph $G$ with $2h+4$ vertices. 
Let $T=\mathrm{CDT}(G)$ be the constrained Delaunay triangulation of $G$. 

We claim that $\varrho(T)=\Theta(\varrho(P))$. We prove this claim in two steps. 
For an intermediate step, we define a polygon with $h$ line segment holes: 
$P'=P_0\setminus \bigcup_{i=1}^h \{a_ib_i\}$. For any point pair $s,t\in P$, denote by 
$\dist(s,t)$ and $\dist'(s,t)$, respectively, the shortest distance in $P$ and $P'$.
Since $P\subseteq P'$, we have $\dist'(s,t)\leq \dist(s,t)$. 
By Lemma~\ref{lem:tech}, $\dist(s,t)< 2.42 \cdot \dist'(s,t)$.
Thus  $\dist'(s,t)=\Theta(\dist(s,t))$ for all $s,t\in P$. 

Every point $s\in P$ lies in one or more triangles in $T$; let $s'$ denote a closest vertex of a triangle in $T$ that contains $s$. For $s,t\in P$, let $\dist''(s,t)$ be the length of the $st$-path $\gamma$ composed of the segment $ss'$, a shortest $s't'$-path in the triangulation $T$, and the segment $t't$. 

Since $\gamma$ does not cross any of the line segments $a_jb_j$, we have 
$\dist'(s,t)\leq \dist''(s,t)$ for any pair of points $s,t\in P$. Conversely, every vertex in the shortest $s't'$-path in $P'$ is an endpoint of an obstacle $a_jb_j$. Consequently, every edge is either an obstacle segment $a_jb_j$, or a visibility edge between the endpoints of two distinct obstacles. By the result of Bose and Keil~\cite{BK06}, for every such edge $pq$, $T$ contains a $pq$-path $\tau_{pq}$ of length $|\tau_{pq}|\leq O(|pq|)$. The concatenation of these paths is an $s't'$-path $\tau$ of length $|\tau|\leq O(\dist'(s',t'))$. Finally, note that the diameter of each triangle in $T$ is at most $\diam_2(P')$. Consequently, if $s,t\in P$ maximizes $\dist(s,t)$, then 
\[
    \dist''(s,t)=|ss'| + |\gamma| +|t't| \leq 2 \cdot \diam_2(P) + |\tau| \leq  O(\dist'(s't')),
\]
as required. Consequently, $\diam_g(T)=\Theta(\diam_g(P))$,
which in turn implies that $\varrho(T)=\Theta(\varrho(P))$.
\end{proof}

\section{Conclusion}
\label{sec:con}

We have shown that in a convex polygonal domain $P$ with $h$ convex holes, the distortion of the geodesic distance, $\varrho(P)=\frac{\diam_g(P)}{\diam_2(P)}$, is always $O(h^{1/2})$ and sometimes $\Omega(h^{1/3})$. Closing the gap between the upper and lower bounds remains an open problem. 
Generalizations to $d$-dimensional Euclidean spaces for $d\geq 3$ are left for future research. 
Improving the running times of algorithms for computing the geodesic diameter or radius of a polygon with $h$ holes when all holes and the outer polygon are convex remains as another interesting problem. 

\paragraph{Acknowledgments.} Research on this paper was partially supported by the NSF awards DMS~1800734 and DMS~2154347.

\bibliographystyle{plainurl}
\bibliography{references.bib}

\end{document}